\documentclass[runningheads]{llncs}
\usepackage{amsmath}
\usepackage{amssymb}
\usepackage[T1]{fontenc}
\usepackage{graphicx}
\usepackage{cite}
\usepackage{geometry}
\usepackage{algorithm2e}
\usepackage{pxfonts}
\usepackage[in]{fullpage}

\def\calP{\mathcal{P}}

\def\Ed{\mathsf{E}\mathrm{d}}

\newtheorem{observation}{Observation}

\begin{document}

\title{The k-Center Problem of Uncertain Points on Graphs\thanks{This research was supported in part by U.S. National Science Foundation under Grant CCF-2339371.}}


\author{Haitao Xu\thanks{Corresponding author} \and Jingru Zhang}

\institute{Cleveland State University, Cleveland, Ohio 44115, USA \\
\email{h.xu12@vikes.csuohio.edu, j.zhang40@csuohio.edu}}
\maketitle             

\begin{abstract}
 In this paper, we study the $k$-center problem of uncertain points on a graph. Given are an undirected simple graph $G = (V, E)$ and a set $\calP$ of $n$ uncertain points where each uncertain point with a non-negative weight has $m$ possible locations on $G$ each associated with a probability. The problem aims to find $k$ centers (points) on $G$ so as to minimize the maximum weighted expected distance of uncertain points to their own expected closest centers. No previous work exist for the $k$-center problem of uncertain points on undirected graphs. We propose exact algorithms that solve respectively the case of $k=2$ in $O(|E|^2m^2n\log |E|mn\log mn )$ time and the problem with $k\geq 3$ in $O(\min\{|E|^km^kn^{k+1}k\log |E|mn\log m, |E|^kn^\frac{k}{2}m^\frac{k^2}{2}\log |E|mn\})$ time, provided with the distance matrix of $G$. In addition, an $O(|E|mn\log mn)$-time algorithmic approach is given for the one-center case. 
\end{abstract}

\section{Introduction}\label{sec:intro}
Facility location problems on uncertain demands have drawn growing attentions~\cite{ref:ChauUn06,ref:HuangSt17,ref:WangAn16,ref:WangCo19,ref:AlipourAp20,ref:AlipourIm21} in academia due to the inherent uncertainty of collected demand data. In this paper, we study the network $k$-center problem, a classical problem in facility locations, for uncertain points where each uncertain point models a demand customer with multiple possible locations on the network. 

Let $G = (V, E)$ be an undirected graph which contains no loops and no multiple edges, and where each edge is of a positive length. For any two vertices $u,v$ of $G$, denote their edge by $e(u,v)$ and the length of $e(u,v)$ by $\overline{e(u,v)}$. Assume $G$ is embedded in the Euclidean plane, so we can talk about ``points'' on an edge: A (deterministic) point $p$ on $e(u,v)$ is represented by a tuple $(u, v, t(p))$, where $t(p)$ is the distance of $p$ to $u$ along $e(u,v)$. For any two points $p$ and $q$ of $G$, the distance $d(p,q)$ between them is the sum of the lengths of all edges on their shortest path $\pi(p, q)$ along $G$. 

An uncertain point of $G$ is a point whose location in $G$ is not certain and follows a discrete probability density function. Let $\calP$ be a set of $n$ uncertain points $P_1, \cdots, P_n$ in $G$. Each $P_i$ is associated with $m$ points $p_{i1}, \cdots, p_{im}$ on $G$, where each $p_{ij}$ is attached with a value $f_{ij}\geq 0$ so that $P_i$ appears at point $p_{ij}$ with the probability $f_{ij}$. Additionally, each $P_i\in\calP$ is of weight $w_i\geq 0$. For any point $p$ of $G$, the distance of $P_i$ to $p$ is defined as $\sum_{j=1}^{m}f_{ij}\cdot d(p_{ij}, p)$, denoted by $\Ed(P_i, p)$, which is the \textit{expected} version of the distance between $P_i$ and $p$. 

The $k$-center problem is to find a set $Q$ of $k$ points (centers) in $G$ so that the maximum weighted expected distance of uncertain points in $\calP$ to their own expected closest points in $Q$ is minimized. Namely, it aims to find a set $Q$ of $k$ points in $Q$ so as to minimize the objective value $\phi(\calP,Q) = \max_{P_i\in\calP} \min_{q\in Q}\{w_i\Ed(P_i,q)\}$. 

When $G$ is a path, the problem can be addressed in $O(mn\log mn + n\log k\log n)$ time~\cite{ref:WangOn15}. For $G$ being a tree graph, this problem can be solved in $O(n^2\log n\log mn+ mn\log^2mn\log n)$ time~\cite{ref:WangCo19}. 
On a general graph, the problem is NP-hard due to the NP-hardness of the case $m=1$~\cite{ref:KarivAnC79,ref:MegiddoNe83}. To the best of our knowledge, no known algorithms exist for solving it. In this paper, we give the first exact algorithms for solving the problem respectively under $k=2$ and $k\geq 3$, and develop a faster algorithm for the case of $k=1$.

\paragraph{Related Work.} For deterministic points, i.e., $m=1$, the network $k$-center problem has been explored a lot in the literature. Kariv and Hakimi~\cite{ref:KarivAnC79} proposed the first exact algorithm that solves the problem on a vertex-weighted graph in $O((|E|^k \cdot n^{2k-1}/(k-1)!)\cdot\lg n )$. Later, Tamir~\cite{ref:TamirIm88} gave a faster exact algorithm that solves the problem in $O(m^k n^k \alpha(n) \log^2 n)$ time where $\alpha(n)$ is the inverse Ackermann function. 
The state-of-art result was given by Bhattacharya and Shi~\cite{ref:BhattacharyaIm14}, and their algorithm computes the exact solution in $O(m^k n^{k/2} 2^{\log^*n} \log n)$ time by remodeling the \textit{decision} problem to a Klee's measure problem, which asks for the measure of the union of axis-parallel boxes. Moreover, faster exact algorithms~\cite{ref:MegiddoNe83,ref:MegiddoLi84,ref:Ben-MosheAn06,ref:Ben-MosheEf07,ref:BhattacharyaOp07,ref:ChenEf15,ref:WangAn21} were proposed for the problem on special graphs including tree and cactus graphs or the problem with a constant $k$.  

When it comes to uncertain points each with multiple possible locations, as the above mentioned, the $k$-center problem was studied first on the path network by Wang and Zhang~\cite{ref:WangOn15}. Later, they studied the problem on tree networks and gave an $O(n^2\log n\log mn+ mn\log^2mn\log n)$-time algorithm~\cite{ref:WangCo19}. In addition, the problem was considered for a constant $k$ respectively on a path, a tree, and a cactus graph, and faster algorithms were developed~\cite{ref:HuCo22,ref:XuTh23,ref:xu2023two,ref:XuTh25}. 
On general graphs, as far as we are aware, no known work exists for the problem even with $k=1$. Also, see the algorithm works~\cite{ref:WangOn15,ref:WangAn16,ref:HuCo22} for the path $k$-center problem of uncertain points with appearance of line segments.   

\paragraph{Our Approach.} 
We observe that the optimal objective value $\lambda^*$ is of the form $w_i\Ed(P_i,x) = w_j\Ed(P_j,x)$ or equals to the $y$-coordinate of a \textit{breakpoint} of some function $y=w_i\Ed(P_i,x)$ for $x$ in some edge. We thus form $y = w_i\Ed(P_i,x)$ of each $P_i\in\calP$ with respect to every edge so that a set of $O(|E|mn)$ lines is obtained such that $\lambda^*$ is decided by their lowest intersection with a \textit{feasible} $y$-coordinate. 

To find $\lambda^*$ among those intersections, we need to address the feasibility test: Given any value $\lambda>0$, determine whether $\lambda\geq\lambda^*$. Two algorithms are proposed for solving this problem. The first approach is based on an essential observation: A set of $O(|E|mn)$ points exists such that there must be $k$ points in this set that can \textit{cover} $\calP$ under $\lambda$ iff $\lambda\geq\lambda^*$. Hence, we test every $k$ points of this set for determining the feasibility of $\lambda$. The second approach, similar to the feasibility test in~\cite{ref:BhattacharyaIm14} for the case $m=1$, determines for every $k$-size subset of $V\cup E$, whether each edge in the subset does contain a point so that these points in edges and vertices in the subset can cover $\calP$ under $\lambda$. This can be reduced into a Klee's measure problem that asks for the measure of the union of axis-parallel boxes, and hence, we can apply the algorithm for the Klee's measure problem to solve it. 

For the case $k=1$, a faster algorithm is obtained by computing the \textit{local} center of $\calP$ in every edge. This can be reduced into a geometry problem that 
finds the lowest point on the upper envelope of polygonal chains. Hence, we can employ the dynamic convex hull data structure~\cite{ref:BrodalDy02} to solve this geometry problem so that the local center of $\calP$ in each edge is obtained in $O(mn\log mn)$ time. 

\section{Preliminaries}\label{sec:pre}
Since the problem on a tree network has been addressed in previous work, we consider only the general situation where $G$ is an undirected simple graph but not a tree. Hence, $|V| = O(|E|)$. 

Provided with the distance matrix, for any two point $p,q$ of $G$, their distance $d(p,q)$ can be obtained in $O(1)$ time since the shortest path length between any two vertices of the edges containing $p$ and $q$ can be known in $O(1)$ time by the distance matrix. As a result, for any $P_i\in\calP$, the expected distance $\Ed(P_i,p)$ of $P_i$ to point $p$ can be computed in $O(m)$ time. 

Recall that $e(u,v)$ is an arbitrary edge of $G$ which is incident to vertex $u$ and $v$. We abuse $e(u,v)$ to the set of all points on it. We say that a point is in the \textit{interior} of edge $e(u,v)$ if it is in set $e(u,v)\setminus\{u,v\}$. For point $p$ and edge $e(u,v)$, we say that a point of $e(u,v)$ is a \textit{semicircular} point of $p$ if this point has two shortest paths to $p$ such that the two paths or their subpaths form a cycle of $G$ including $e(u,v)$.

Recall that $x = (u, v, t(x))$ is a point of $e(u,v)$ at distance $t(x)$ to $u$ along $e(u,v)$. For any $P_i\in \calP$, $\Ed(P_i,x)$ is a function in $t(x)$ as $x$ moves along $e(u,v)$ from $u$ to $v$. We have the following lemma for $\Ed(P_i,x)$. 

\begin{lemma}\label{lem:expecteddistanceproperty}
    For any $P_i\in \calP$, function $\Ed(P_i,x)$ for $x\in e(u,v)$ is a piecewise linear function consisting of $O(m)$ pieces, and it can be explicitly computed in $O(m\log m)$ time. 
\end{lemma}
\begin{proof}
Let $p$ be any point of $G$. If $p\in e(u,v)$, suppose $p$ is at distance $t(p)$ to $u$ along $e(u,v)$. We first consider the distance $d(p,x)$ from $p$ to $x$. 

On the one hand, $p$ is not in set $e(u,v)\setminus\{u,v\}$. 
As illustrated in Fig.\ref{fig:distanceF} (a), as $x$ moves from $u$ to $v$ along $e(u,v)$, $d(p,x)$ increases, decreases, or first increases until $x$ reaches a point of $e(u,v)$ and then starts to decrease~\cite{ref:KarivAnC79}. 
In details, if $d(p,u) + \overline{e(u,v)} = d(p,v) $ (i.e., the first case of $d(p,x)$), then $d(p,x) = d(p,u) + t(x)$ at any $x\in e(u,v)$. 
If $d(p,v) + \overline{e(u,v)} = d(p,u) $ (i.e., the second case of $d(p,x)$), $d(p,x) = d(p,v) + \overline{e(u,v)} - t(x)$ at any $x\in e(u,v)$. 
Otherwise, $d(p,u) + \overline{e(u,v)} > d(p,v) $ and $d(p,v) + \overline{e(u,v)} > d(p,u)$; 
so $p$ has an unique point $p'=(u,v,t(p'))$ in the interior of $e(u,v)$ such that $d(p,u) + t(p') = d(p,v) + \overline{e(u,v)} - t(p')$, but $d(p,u) + t(x) < d(p,v) + \overline{e(u,v)} - t(x)$ for any $x$ in the segment $[u, p')$ of $e(u,v)$ 
and $d(p,u) + t(x) > d(p,v) + \overline{e(u,v)} - t(x)$ for $x\in (p',v]$. 
By the definition, $p'$ is the semicircular point of $p$ on $e(u,v)$. 
In this case, $d(p,x) = d(p,u) + t(x)$ for $x\in [u, p']$, and 
$d(p,x) = d(p,v) + \overline{e(u,v)} - t(x)$ for $x\in [p',v]$. 

 \begin{figure}
 \begin{minipage}{0.5\textwidth}
    \centering
    \includegraphics[scale = 0.7]{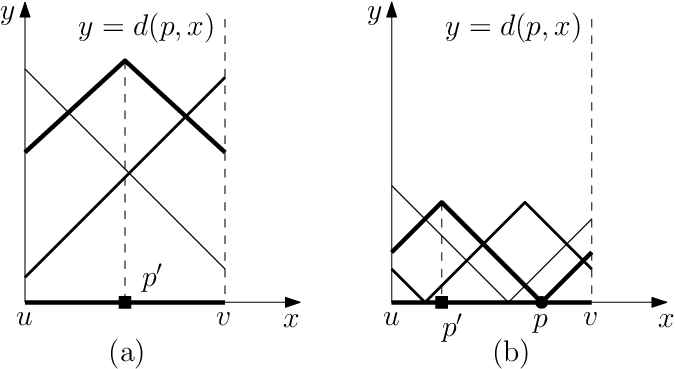}
    \caption{Illustrating function $d(p,x)$ of a point $p$ for $x$ in edge $e(u,v)$. (a) is for the case where $p$ is not in $e(u,v)\setminus\{u,v\}$ while (b) is for the situation where $p$ is in the interior of $e(u,v)$. 
    }
    \label{fig:distanceF}        
 \end{minipage}
 \begin{minipage}{0.5\textwidth}
    \centering
    \includegraphics[scale = 0.57]{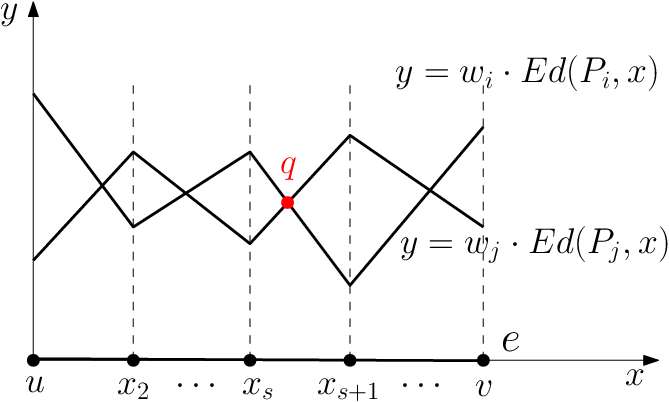}
    \caption{Illustrating that the local center of two uncertain points $P_i$ and $P_i$ on edge $e(u,v)$ is determined by the lowest intersection $q$ of their weighted expected distance functions where their slopes are of opposite signs.}
    \label{fig:Edfunction}
\end{minipage}
 \end{figure}

On the other hand, $p$ is in the interior of $e(u,v)$. As illustrated in Fig.\ref{fig:distanceF} (b), 
if $d(p,u)<t(p)$, i.e, $\pi(p,u)\notin e(u,v)$, then there is an unique point $p'$ in the interior 
of $e(u,v)$ so that $d(p,u) + t(p') = t(p) - t(p')$, but $d(p,u) + t(x) < t(p) - t(x)$ for any $x\in [u, p')$ and 
$d(p,u) + t(x) > t(p) - t(x)$ for $x\in (p',v]$. Clearly, $p'$ is the semicircular point of $p$ on $e(u,v)$. 
Hence, if $d(p,u)<t(p)$, $d(p,x) = t(x) + d(p,u)$ for $x\in [u, p']$ and $d(p,x) = |t(p) - t(x)|$ for any $x\in [p', v]$. Likewise, if $d(p,v)<\overline{e(u,v)} - t(p)$, then there is an unique point $p''=(u,v,t(p''))$ in $e(u,v)$, which is between $p$ and $v$, so that $t(p'') - t(p) = d(p,v) + \overline{e(u,v)}- t(p'')$, 
but $|t(x) - t(p)| < d(p,v) + \overline{e(u,v)}- t(x)$ for any $x\in [u,p'')$ 
and $t(x) - t(p) > d(p,v) + \overline{e(u,v)}- t(x)$ for any $x\in (p'',v]$. Point $p''$ is $p$'s semicircular point on $e(u,v)$.  
Otherwise, $d(p,u) =t(p)$ and $d(p,v)=\overline{e(u,v)} - t(p)$, so $d(p,x) = |t(p) - t(x)|$ for any $x\in e(u,v)$. 

Thus, $d(p,x)$ for $x\in e(u,v)$ is a piecewise linear function consisting of $O(1)$ pieces; it turns at $p$'s semicircular point and $p$ if $p$ has a semicircular point on $e(u,v)$ and $p\in e(u,v)$. Refer to any point on $e(u,v)$ where $d(p,x)$ 
turns as its \textit{turning} point. It follows that $\Ed(P_i,x)$ is a piecewise linear function, which turns at each turning point of functions $d(p_{ij},x)$ for all $1\leq j\leq m$, 
and its graph consists of $O(m)$ pieces. 

Now we discuss how to determine $\Ed(P_i,x)$ for $x\in e(u,v)$. First, we determine $d(p_{ij}, x)$ 
in $x\in e(u,v)$ for all $1\leq j\leq m$ and all their turning points on $e(u,v)$ in $O(m)$ time. Then, sort all turning points on $e(u,v)$ by their distances to $u$ along $e(u,v)$ in $O(m\log m)$ time. Join $u$ and $v$ to the ordered set of turning points. We continue to scan the set to remove duplicates while compute for each distinct turning point, a list that contains all locations of $P_i$ whose distance functions $d(p_{ij},x)$ turn at this turning point. Let $X_i = \{x_{i,1}, \cdots, x_{i,m_i}\}$ be the obtained (ordered) set of $\Ed(P_i,x)$'s turning points on $e(u,v)$, where $m_i = O(m)$, $x_{i,1} = u$, and $x_{i,m_i} = v$. Denote by $L_{i,s}$ the obtained list for $x_{i,s}$ which contains all $P_i$'s locations whose functions $d(p_{ij},x)$ turn at $x_{i,s}$. Set $X_i$ can be obtained totally in $O(m\log m)$ time. 

Proceed to determine $\Ed(P_i,x)$ over each segment $[x_{i,s}, x_{i,s+1}]$ for all $1\leq s\leq m_i-1$ in order. Loop $X_i$ from the beginning. 
For $x_{i1}$, we sum $f_{ij}d(p_{ij}, x)$ at $x = u$ for all $1\leq j\leq m$ to determine $\Ed(P_i,x)$ over $[x_{i,1}, x_{i,2}]$, and attach the obtained function with $x_{i,1}$ in $X_i$. For each $x_{i,s}$ with $2\leq s < m_i$, letting $g(x)$ represent $\Ed(P_i,x)$ over $[x_{i,s-1}, x_{i,s}]$ (which is a linear function), we loop through the list $L_{i,s}$ of $x_{i,s}$ to determine $\Ed(P_i,x)$ over $[x_{i,s}, x_{i,s+1}]$ as follows. For each location $p_{ij}$ in $L_{i,s}$, subtract $f_{ij}d(p_{ij},x)$ for $x\in [x_{i,s-1}, x_{i,s}]$ from $g(x)$ and then add $f_{ij}d(p_{ij},x)$ for $x\in [x_{i,s}, x_{i,s+1}]$ to $g(x)$. 
Recall that $d(p_{ij},x)$ at any point can be determined in $O(1)$ time. It thus takes $O(|L_{i,s}|)$ time to determine $\Ed(P_i,x)$ for $x\in [x_{i,s}, x_{i,s+1}]$. Due to $\sum_{s=1}^{s=m_i}L_{i,s} = O(m)$, 
it takes $O(m\log m)$ to determine $\Ed(P_i,x)$ for $x\in e(u,v)$, provided with the distance matrix. 

Thus, the lemma holds. \qed
\end{proof}

Lemma~\ref{lem:expecteddistanceproperty} computes a set $X_i = \{x_{i1}, \cdots, x_{i,m_i}\}$ for $P_i$ that consists all turning points of function $\Ed(P_i,x)$ on $e$ where each point $x_{i,s}$ is attached with $\Ed(P_i,x)$ only for $x\in [x_{i,s}, x_{i,s+1}]$. Additionally, $(x_{i,j}, \Ed(P_i,x_{i,j}))$ for all $1\leq j\leq m_i $ are all breakpoints of $y = \Ed(P_i,x)$. 

Recall that $\lambda^*$ is the optimal objective value of the problem. Although $\Ed(P_i,x)$ may be neither convex nor concave for an edge of $G$, the following Lemma leads a finite candidate set for $\lambda^*$. (See Fig.\ref{fig:Edfunction}.) 

\begin{lemma}\label{lem:candidateset}
    $\lambda^*$ is of the form $w_i\cdot\Ed(P_i,x) = w_j\cdot\Ed(P_j,x)$ so that the slopes of the two functions at the intersection are of opposite signs, or $\lambda^*$ is determined by a breakpoint of some function $w_i\cdot\Ed(P_i,x)$. 
\end{lemma}
\begin{proof}
    Suppose that $\lambda^*$ falls into neither of the two cases. 
    Let $x'$ be a center that determines $\lambda^*$, and let $\calP'$ 
    be the subset of uncertain points assigned to $x'$. 
    If $x'$ is in the interior of an edge, e.g., $e(u,v)$, 
    then one can move $x'$ along $e(u,v)$ toward one vertex to reduce the maximum weighted expected distance $\phi(\calP',x')$ of $x'$ to $\calP'$ 
    until $x'$ reaches a point $x''$ so that $(x'', \phi(\calP',x''))$ is either an intersection of the weighted expected distance functions of two uncertain points in $\calP'$ where their slopes are of opposite signs, or a breakpoint of the weighted expected distance function of an uncertain point in $\calP'$, 
    which leads a contradiction with the assumption. Thus, the lemma holds. \qed
\end{proof}

In the following, we assume that for each edge of $G$, functions $\Ed(P_i,x)$ of all uncertain points, i.e., their sets $X_i$, have been determined in the preprocessing work, which takes $O(|E|mn\log m)$ time. Additionally, 
when we talk about points in $d$-dimension space, 
for any point $q$ on an axis, we also use $q$ to 
represent its coordinate along that axis.


\section{The algorithm}\label{alg:general}
We first present in Subsection~\ref{subsec:decision} our feasibility test that decides 
whether $\lambda\geq\lambda^*$ for any given $\lambda>0$, and then give 
the algorithmic approach for finding $\lambda^*$ in Subsection~\ref{subsec:computinglambda}. 

\subsection{The feasibility test}\label{subsec:decision}
Given any value $\lambda>0$, the goal is to decide whether there is 
a set $Q$ of $k$ points on $G$ so that $\phi(Q,\calP)\leq\lambda$. 
If yes then $\lambda\geq\lambda^*$ and $\lambda$ is feasible. 
Otherwise, $\lambda<\lambda^*$ so $\lambda$ is infeasible. 

We say that an uncertain point is \textit{covered} by a point under $\lambda$ if its (weighted) expected distance to that point is no more than $\lambda$. 
Hence, $\lambda^*$ is the smallest feasible value so that there are $k$ points in $G$ such that each $P_i\in\calP$ is covered by one of them under $\lambda^*$. 
We give two algorithms to decide the existence of such $k$ points in $G$ 
for any given $\lambda>0$. 

\paragraph{The first algorithm.} For each edge $e$, consider solving $w_i\Ed(P_i,x)\leq\lambda$ for each $1\leq i\leq n$. This generates for each $P_i$ $O(m)$ disjoint intervals, called feasible intervals, on $e$. Let $Q(E)$ be the union of $V$ and endpoints of the feasible intervals of uncertain points on all edges. Clearly, $|Q(E)| = O(|E|mn)$. We have the following observation.



\begin{observation}\label{obs:candidateset}
    If $\lambda$ is feasible, then there exist $k$ points in $Q(E)$ that can cover $\calP$ under $\lambda$.  
\end{observation}
\begin{proof}
    Let $Q'$ be any $k$-point set that covers $\calP$ under $\lambda$. Assume a point $q$ of $Q'$ is not in $Q(E)$, and denote by $\calP'$ the subset of uncertain points covered by $q$. Clearly, $q$ is in the interior of a feasible interval of each $P_i\in\calP'$ on some edge, and the intersection containing $q$ of these feasible intervals is an interval. Replacing $q$ with any endpoint 
    of this intersection leads a $k$-point set belonging to $Q(E)$ that covers $\calP$. This proves the lemma. \qed
\end{proof}

The above observation guides the first approach: Compute $Q(E)$ and 
for every $k$ distinct points of $Q(E)$, determine whether they can cover $\calP$ 
under $\lambda$. If no $k$ points of $Q(E)$ cover $\calP$, then $\lambda$ is in feasible. 

Recall that for each edge $e$ and each $P_i\in\calP$, the ordered set $X_i$ obtained in the preprocessing contains $\Ed(P_i,x)$'s turning points on $e$, and each point in $X_i$ is associated with $\Ed(P_i,x)$'s function over the interval between it and its next point. Consequently, it takes $O(m)$ time to solve $w_i\Ed(P_i,x)\leq\lambda$ for $x\in e$, and given any point $p$, value $w_i\Ed(P_i,p)$ can be known in $O(\log m)$ time by a binary search on $X_i$.

Thus, $Q(E)$ can be obtained in $O(|E|mn)$ time. For any $k$ points, we can determine whether they can cover $\calP$ in $O(kn\log m)$ time. Hence, the feasibility of $\lambda$ can be determined in $O(|E|^km^kn^{k+1}k\log m)$ time. We thus have the following result. 
\begin{lemma}\label{lem:approach1}
The feasibility of $\lambda$ can be determined in $O(|E|^km^kn^{k+1}k\log m)$ time. 
\end{lemma}

\paragraph{The second algorithm.}Given a set $E'$ of $k'$ edges 
$e_1, \cdots, e_{k'}$ and a set $V'$ of $k''$ distinct vertices 
so that $k' + k'' =k$, the \textit{local} feasibility test is 
to determine whether there are $k'$ points $q_1, \cdots, q_{k'}$ so that 
for each $1\leq s\leq k'$, $q_s$ is in the interior of $e_s$ and $V'\cup\{q_1,\cdots, q_{k'}\}$ cover $\calP$ under $\lambda$. Similar to the algorithm~\cite{ref:BhattacharyaIm14} for $m=1$, we solve the local feasibility test for every $k$-size subset of $E\cup V$ that however, allows for multiple instances for an edge for addressing the situation where an edge may contain multiple centers under $\lambda$. 

For any given subset $E'\cup V'$, we first determine all uncertain points covered by vertices of $V'$, which can be found in $O(kn)$ time since the expected distance of each $P_i\in\calP$ at a vertex (the first or last point in set $X_i$ for an edge) can be known in $O(1)$ time. Next, we determine whether each edge of $E'$ contains a point so that these $k'$ points cover all uncovered uncertain points, which is a special case of the local feasibility test. If yes, then $\lambda$ is feasible, and otherwise, $\lambda$ is infeasible w.r.t. $E'\cup V'$. 

Due to $|V| =O(|E|)$, the algorithm runs in $O(|E|^k (kn + \tau))$, where $\tau$ is the running time for solving the above special case of the local feasibility test. 

We now show in the below how to address the local feasibility test with only edge input: Given a $k$-size multiset of edges $e_1, \cdots, e_k$ in $E$, it asks whether there are $k$ points $q_1, \cdots, q_k$ with each $q_s$ in $e_s$ excluding its endpoints so that they cover $\calP$ under $\lambda$. 

Similar to the algorithm in~\cite{ref:BhattacharyaIm14}, we transform this problem into a Klee's measure problem: Given $M$ axis-parallel boxes in $d$-dimension space, it asks for the volume of the union of these boxes. 
For $d =2$, this measure can be computed in $O(M\log M)$ time~\cite{ref:BentleyAl77,ref:ChazelleOn96} while for $d\geq 3$, this measure can be calculated in $O(M^\frac{d}{2})$ time by Chan's algorithm~\cite{ref:ChanKl13}. 


\paragraph{The local feasibility test on edges.} Consider each edge $e_s$ as a line segment $(0, \overline{e_s})$ on the $s$-th axis in $k$-dimension space. Also, consider the slab of the form $\{(x_1, \cdots, x_k)| 0< x_s<\overline{e_s}\}$. 
The intersection of the slab for all $1\leq s\leq k$ is an axis-parallel box $\Gamma$ in $k$-dimension space; in other words, $\Gamma$ is the cartesian product of segments $(0, \overline{e_s})$ for all $1\leq s\leq k$. 

For any $1\leq i\leq n$, solving $w_i\Ed(P_i,x)>\lambda$ for each $1\leq s\leq k$ generates a set $I^s_i$ of $O(m)$ disjoint open segments $I_{i,1}^s, \cdots, I_{i,m}^s$ on segment $(0, \overline{e_s})$, where each $I_{i,j}^s = (a_{i,j}^t, b_{i,j}^t)$ is called a \textit{forbidden} segment of $P_i$ on $e_s$. Note that if set $I_i^s$ is empty then we prune $P_i$ for this local feasibility 
test in that any point of $e_s$ covers $P_i$. 

Likewise, for each $P_i\in\calP$, as illustrated in Fig.\ref{fig:3dbox}, the cartesian product of its forbidden segments leads a set $B_i$ of $m^k$ axis-parallel boxes in $k$-dimension space each of the form $\{(x_1, \cdots, x_k)| x_1\in I_{i,j_1}^1, \cdots, x_k\in I_{i,j_k}^k\}$. 

\begin{figure}
    \centering
    \includegraphics[width=0.38\linewidth]{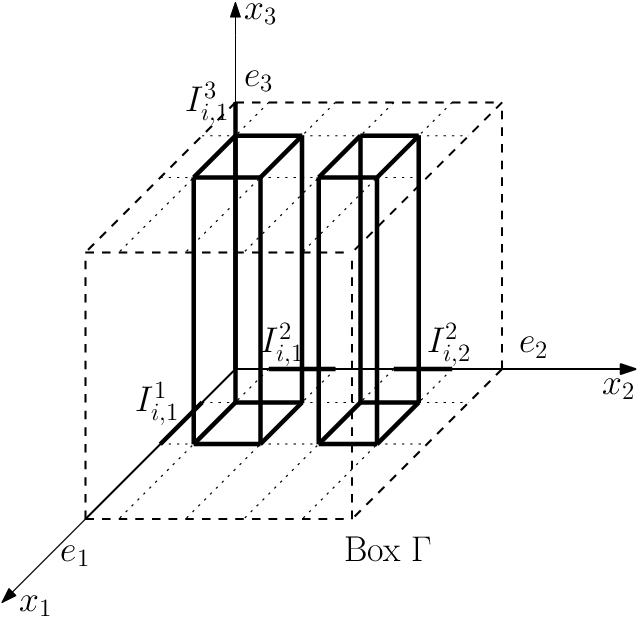}
    \caption{Illustrate that an uncertain point $P_i$ contains one forbidden segment $I^1_{i,1}$ on $e_1$, 
    two forbidden segments $I^2_{i,1}, I^2_{i,2}$ on $e_2$, and one forbidden segment $I^3_{i,1} = e_3$ on $e_3$; the cartesian product of its forbidden segments on $e_1, e_2$ and $e_3$ leads two boxes $\{x_1, x_2, x_3| x_1\in I^1_{i,1}, x_2\in I^2_{i,1}, x_3\in I^3_{i,1}\}$ and $\{x_1, x_2, x_3| x_1\in I^1_{i,1}, x_2\in I^2_{i,2},x_3\in I^3_{i,1}\}$ contained in box $\Gamma$ in $3$-dimension space.}
    \label{fig:3dbox}
\end{figure}

Consider any $k$ points $q_1, \cdots, q_k$ with $q_s\in (0, \overline{e_s})$ for each $1\leq s\leq k$ in $k$-dimension space. Let $q$ be the point whose $s$-th coordinate equals to that of $q_s$, so $q$ is a point in box $\Gamma$. 
Observe that if a box of $B_i$ contains $q$ then $P_i$ cannot be covered by any of the $k$ points. 
It follows that if for every $k$ points 
$q_1, \cdots, q_k$ with $q_s\in (0, \overline{e_s})$ for each $1\leq s\leq k$, the corresponding point $q$ in $k$-dimension space is in a box of set $\cup_{i=1}^{n}B_i$, then $\lambda$ is infeasible w.r.t. the given $k$ edges (excluding their endpoints). This implies the following observation. 

\begin{observation}\label{lem:LocalInfeasibility}
    $\lambda$ is infeasible w.r.t. $e_1, \cdots, e_k$ (excluding their endpoints) iff $\cup_{i=1}^{n} B_i$ covers $\Gamma$. 
\end{observation}

Clearly, if the volume of $\cup_{i=1}^{n} B_i$ equals to $\Gamma$'s, then $\lambda$ is infeasible w.r.t. $e_1, \cdots, e_k$, and $\lambda$ is feasible, otherwise. Computing the volume of $\cup_{i=1}^{n} B_i$ is a Klee's measure problem. However, since each box in $\cup_{i=1}^{n} B_i$ is a box whose edges are open segments, we cannot decide if $\cup_{i=1}^{n} B_i$ covers $\Gamma$ by comparing their volumes. 
Instead, a closed box is constructed for each box in set $B_i$ of each $P_i\in\calP$ such that $\lambda$ is infeasible if the volume of their union equals to 
that of $\Gamma$. 

For each segment $e_s$, let $\delta_s$ be the minimum distance larger than zero between endpoints of all forbidden segments on $e_s$. Let $\delta = \frac{1}{2}\min_{1\leq s\leq k}\delta_s$. For each forbidden segment $I_{i,j}^s = (a_{i,j}^s, b_{i,j}^s)$ of $P_i$, 
consider the segment $L_{i,j}^s = [\alpha_{i,j}^s, \beta_{i,j}^s]$ on $e_s$ following the rules: 
If $a_{i,j}^s = 0$ then $\alpha_{i,j}^s = a_{i,j}^s$, 
and otherwise, $\alpha_{i,j}^s = a_{i,j}^s +\delta$; if $b_{i,j}^s = \overline{e_s}$ then $\beta_{i,j}^s = b_{i,j}^s$, and otherwise, $\beta_{i,j}^s = b_{i,j}^s -\delta$. For each $P_i\in\calP$, the cartesian product of its new segment set $\{L_{i,1}^s, \cdots, L_{i,m}^s\}$ for all $1\leq s\leq k$ leads a set $D_i$ of $m^k$ axis-parallel boxes each of the form $\{(x_1, \cdots, x_k)| x_1\in L_{i,j_1}^1, \cdots, x_k\in L_{i,j_k}^k\}$ 
in $k$-dimension space. Last, let $\Gamma'$ be the box $\Gamma$ including its vertices. 

For $m=1$, each set $D_i$ contains only one box in $k$-dimension space since $P_i$ has only one forbidden segment on each edge; it was proved in~\cite{ref:BhattacharyaIm14} that $\cup_{i=1}^n D_i$ covers $\Gamma'$ if and only if $\cup_{i=1}^n B_i$ covers $\Gamma$. Because each $D_i$ is obtained by shrinking each edge of $B_i$ from both sides by value $\delta$, which is half of the minimum distance between parallel faces of all boxes. That observation applies to our case where each $D_i$ contains multiple axis-parallel boxes in $k$-dimension space. We thus have the following lemma. 

\begin{lemma}\label{lem:boxcoverage}
    $\cup_{i=1}^n D_i$ covers $\Gamma'$ if and only if $\cup_{i=1}^n B_i$ covers $\Gamma$. 
\end{lemma}

Now we are ready to describe the algorithm for the local feasibility test. Compute the universe box $\Gamma'$ in $O(k)$ time by determining the coordinates of all endpoints of segments $e_1, \cdots e_k$. Next, for each $1\leq i\leq n$, we solve $w_i\Ed(P_i, x)>\lambda$ w.r.t. each $e_s$ to find $O(m)$ forbidden segments of $P_i$, which can be done totally in $O(kmn)$ time. 
We prune every uncertain point that has no forbidden segments on an edge. 
In case no uncertain points remain, we immediately return that $\lambda$ is feasible. 

Proceed to compute $\delta$ in $O(kmn\log mn)$ time by spending $O(mn\log mn)$ time on finding the closest pair of endpoints of distance larger than zero for all forbidden segments on each edge. 
Next, we construct the set $D_i$ of $O(m^k)$ boxes for each $P_i\in\calP$ as the above described, 
which can be carried out totally in $O(knm^k)$ time. 

Further, we compute the volume of $\cup_{i=1}^{n}D_i$: For $k=2$, this measure can be obtained in $O(nm^2\log nm^2)$ time~\cite{ref:BentleyAl77,ref:ChazelleOn96}; for $k\geq 3$, it can be computed in $O(n^\frac{k}{2}m^\frac{k^2}{2})$ time~\cite{ref:ChanKl13}. Last, by comparing their volumes, we can decide whether $\cup_{i=1}^{n}D_i$ 
covers $\Gamma$. 

Hence, we can determine the feasibility of $\lambda$ w.r.t. any $k$ edges of $E$ (excluding the endpoints) in $O(nm^2\log mn)$ time for $k=2$ and in $O(n^\frac{k}{2}m^\frac{k^2}{2})$ time for $k\geq 3$.  

\paragraph{Wrapping things up.} 
Recall that the second approach solves the feasibility test in $O(|E|^k \cdot(kn+\tau))$ time, where $\tau$ is the running time of the local feasibility test on $k$ input edges and it follows the above result. 

Combining with Lemma~\ref{lem:approach1}, we have the following result for the feasibility test. 
Notice that for $k>3$, the first algorithm is more efficient when $m \geq n$. 

\begin{lemma}\label{lem:approach2}
For $k=2$, the feasibility of $\lambda$ can be determined in $O(|E|^2m^2n\log mn)$ time; for $k\geq 3$, the feasibility of $\lambda$ can be determined in $O(\min\{|E|^km^kn^{k+1}k\log m, |E|^km^\frac{k^2}{2}n^\frac{k}{2}\}\} )$ time. 
\end{lemma}


\subsection{Computing $\lambda^*$}\label{subsec:computinglambda}

Lemma~\ref{lem:candidateset} implies that $\lambda^*$ is of the form either $w_i\Ed(P_i,x) = w_j\Ed(P_j,x)$ or decided by a breakpoint of $w_i\Ed(P_i,x)$ for some $1\leq i\leq n$. To find $\lambda^*$, we implicitly enumerate this candidate set of $\lambda^*$ as the set of $y$-coordinates of intersections between lines containing line segments on graphs of functions $y = w_i\Ed(P_i,x)$, so that we can employ the line arrangement search technique~\cite{ref:ChenAn13} 
to find $\lambda^*$ with the assistance of our feasibility test. 

Consider the $x,y$-coordinate system where every edge $e$ is 
a line segment $[0, \overline{e}]$ on the $x$-axis. 
Consider for each $1\leq i\leq n$ function $y = w_i\Ed(P_i,x)$ w.r.t. each edge. 
Let $L$ be the set of lines containing line segments on graphs of 
functions $y=w_i\Ed(P_i,x)$ and vertical lines through $e$'s endpoints 
on the $x$-axis. Clearly, $\lambda^*$ is the $y$-coordinate of the lowest intersection 
between lines in $L$ with a feasible $y$-coordinate. $|L| = O(|E|mn)$ and $L$ can be 
obtained in $O(|E|mn)$ time (since $\Ed(P_i,x)$ for each edge has been determined in the preprocess work). 

Denote by $\mathcal{A}_L$ the arrangement of lines in $L$. 
In $\mathcal{A}_L$, every intersection of lines defines a vertex of $\mathcal{A}_L$ and vice versa. 
Let $v_1$ be the lowest vertex of $\mathcal{A}_L$ whose $y$-coordinate $y_{v_1}$ is a feasible value, 
and let $v_2$ be the highest vertex of $\mathcal{A}_L$ whose $y$-coordinate $y_{v_2}$ is smaller 
than $y_{v_1}$. By the definitions, $y_{v_2}<\lambda^*\leq y_{v_1}$ and no vertices in $\mathcal{A}_L$ 
have $y$-coordinates in range $(y_{v_2}, y_{v_1})$. Clearly, $\lambda^* = y_{v_1}$. 

Lemma 2 in~\cite{ref:ChenAn13} can be utilized to find the two vertices without constructing $\mathcal{A}_L$ 
by employing our Lemma~\ref{lem:approach2} to decide the feasibility of the $y$-coordinate 
for every tested vertex of $\mathcal{A}_L$. The time complexity is $O((|L|+ T)\log |L|)$ 
where $T$ is the running time of our feasibility test. Plus the preprocessing time for constructing functions $\Ed(P_i,x)$, we have the following theorem. 



\begin{theorem}\label{the:kcenter}
    The $k$-center problem can be solved in $O(|E|^2m^2n\log ^2mn + |E|^2m^2n\log |E|)$ for $k=2$, and solved in $O(\min\{|E|^km^kn^{k+1}k\log |E|mn\log m, |E|^kn^\frac{k}{2}m^\frac{k^2}{2}\log |E|mn\})$ time for $k\geq 3$. 
\end{theorem}

\section{Solving the one-center problem}\label{alg:onecenter}
For $k=1$, our strategy is to find the \textit{local} center of $\calP$ 
on each edge of $G$ so that the one with the smallest objective value is their center. 
For each edge $e$, consider function $y = \Ed(P_i,x)$ for each $P_i\in\calP$ 
in the $x,y$-coordinate plane, and define $\phi_e(\calP,x) = \max_{P_i\in\calP}w_i\Ed(P_i,x)$. By Lemma~\ref{lem:expecteddistanceproperty}, $\phi_e(\calP,x)$ is the upper envelope of functions $w_i\Ed(P_i,x)$ (polygonal chains) for $x\in e$ and the local center of $\calP$ on $e$ is decided by the lowest point on their upper envelope (see Fig.\ref{fig:Edfunction}). 

Recall that $X_i = \{x_1, \cdots, x_{m_i}\}$ is the ordered set of 
all turning points of function $y = w_i\Ed(P_i,x)$ on $e$, which 
has been determined in the preprocessing work. 
Let $X = \{1\leq s\leq M\}$ be the union of $X_i$ for all $1\leq i\leq n$ 
and $M=O(mn)$. Clearly, $ y = w_i\Ed(P_i,x)$ is a line segment 
over each interval $[x_s, x_{s+1}]$ with $1\leq s< M$.


Consider the problem of computing the lowest points on the upper envelope 
of functions $w_i\Ed(P_i,x)$'s over interval $ [x_s, x_{s+1}]$ for all $1\leq s < M$ from left to right. 
Because the upper envelope of lines is geometric dual to the convex (lower) hull of 
points~\cite{ref:BrodalDy02}. This problem can be addressed by applying the dynamic convex-hull maintenance data structure~\cite{ref:BrodalDy02} which supports logarithmic-time insertion, deletion, query about the lowest point on the dual upper envelope, and query about the point where a vertical line intersects the dual upper envelope. The details are presented as follows. 

Merge sets $X_i$ for all $1\leq i\leq n$ into a whole ordered list $X$; 
remove duplicates from $X$ while compute for each point $x_s\in X$ 
a subset $\calP(x_s)$ of uncertain points so that each $P_i\in\calP(x_s)$ 
has its function $\Ed(P_i,x)$ turning at $x_s$ and it is associated with its function $\Ed(P_i,x)$ over $[x_s, x_{s+1}]$. These operations can be carried out in $O(mn\log mn)$ time by sorting sets $X_i$ and then looping $X$ to remove duplicates while computing that subset for each point. Note that $x_1\in X$ is a vertex of $e$ and the other vertex of $e$ is the last point $x_M$ of $X$; additionally, subset $\calP(x_1) = \calP$ and $\calP(x_M) = \phi$. 

Proceed to find the lowest point on the upper envelope $\phi_e(\calP,x)$ over 
intervals $ [x_s, x_{s+1}]$ in order. 
For $x_1\in X$, we construct in $O(n\log n)$ time the dynamic convex-hull maintenance data structure $Y$~\cite{ref:BrodalDy02} for $n$ lines generated by extending the line segment 
of each $y =w_i\Ed(P_i,x)$ for $x\in [x_{1},x_{2}]$. 
Then, we perform in order the following three queries on $Y$ in $O(\log n)$ time: 
the lowest-point query and the queries about the two intersections where vertical lines $x =x_1$ and $x=x_2$ 
intersect the upper envelope. For our purpose, we keep only the lowest one of the obtained points 
in slab $x_1\leq x\leq x_2$. 

Starting with $s=2$, we compute the lowest point of $\phi_e(\calP,x)$ over intervals $[x_s, x_{s+1}]$ 
with $2\leq s< M$ as follows. 
Suppose we are about to compute the lowest point of $\phi_e(\calP,x)$ for $x\in [x_s, x_{s+1}]$. 
At this moment, $Y$ maintains the upper envelope of lines containing line segments of functions 
$y = w_i\Ed(P_i,x)$ over the previous interval $[x_{s-1},x_{s}]$. 
Recall that subset $\calP(x_s)$ contains all uncertain points whose functions $\Ed(P_i,x)$ turn at $x=x_s$. 
For each $P_i\in\calP(x_s)$, we first delete from $Y$ its function over $[x_{s-1}, x_{s}]$ in $O(\log n)$ time, 
and then insert its function over $[x_{s}, x_{s+1}]$, which is associated with $P_i$ in $\calP(x_s)$. 
After these $2|\calP(x_s)|$ updates on $Y$, we find the lowest point of $\phi_e(\calP,x)$ over 
$[x_{s}, x_{s+1}]$ as for interval $[x_1, x_2]$ by performing the three queries for 
interval $[x_s, x_{s+1}]$. Last, find the lowest point of $\phi_e(\calP,x)$ 
over $[x_1, x_{s+1}]$ in $O(1)$ time. 

It is easy to see that the processing time for interval $[x_s, x_{s+1}]$ 
is dominated by the term $(mn\log mn + 2|\calP(x_s)|\log n)$. 
Thus, with $O(mn\log m)$ preprocessing work for determining functions $\Ed(P_i,x)$ for $x\in e$, 
the lowest point of $\phi_e(\calP,x)$, i.e., the local center of $\calP$ on $e$ and its objective value, 
can be found in $O(mn\log mn + \log n\cdot\sum_{s=1}^{M}|\calP(x_s)|)$ time, 
which is $O(mn\log mn)$. 

Apply the above routine to every edge of $E$ to find every local center of $\calP$. 
As a result, with the provided distance matrix, the center of $\calP$ on $G$ can be found 
in $O(|E|mn\log mn)$ time. 


\begin{theorem}
    The one-center problem can be solved in $O(|E|mn\log mn)$ time. 
\end{theorem}

\section{Conclusion}
In this paper, we study the $k$-center problem with $n$ (weighted) uncertain points on an undirected simple graph. We give the first exact algorithms for solving this problem respectively under $k =1$, $k=2$, and $k\geq 3$. 
Our problem generalizes the continuous graph $k$-center problem by assuming that the location of each demand customer follows a discrete distribution, which is quite practical especially for locations with measurement errors or dynamic customers. Since the input of our problem involves the probabilistic points on graphs, the developed algorithms may provide insights for addressing related problems in probabilistic combinatorics. Further, due to the NP-hardness of this problem, it is quite interesting to explore the approximation algorithms for solving it and study the problem on special graphs including the cactus graph, the partial $k$-tree, etc. 


\section*{Statements and Declarations}

\paragraph{CRediT authorship contribution statement.} 
\textbf{Haitao Xu:} Writing - review \& editing, Writing - original draft, Investigation, Formal analysis. \textbf{Jingru Zhang:} Writing - review \& editing, Supervision, Investigation, Funding acquisition, Formal analysis.

\paragraph{Declaration of Competing Interests.} The authors have no relevant financial or non-financial interests to disclose. 



\paragraph{Data availability.} This research did not use or generate any data set.

\paragraph{Funding.} This work was supported by U.S. National Science Foundation under Grant CCF-2339371. 

\end{document}